\documentclass[a4paper,11pt]{article}

\usepackage{amsmath,amsthm}
\usepackage{amsfonts}
\usepackage{amssymb}

\usepackage{color}

\theoremstyle{plain}
\newtheorem{theorem}{Theorem}[section]

\newtheorem{proposition}[theorem]{Proposition}

\theoremstyle{definition}
\newtheorem{definition}{Definition}

\theoremstyle{remark}
\newtheorem{remark}{Remark}

\title{{
On the weak order ideal associated to linear codes}}

\author{Mijail Borges-Quintana, \and~{Miguel~\'Angel Borges-Trenard}\thanks{M.~Borges-Quintana and M.A.~Borges-Trenard are with the Department
of Department of Mathematics, Faculty of Ciencias Naturales y Exactas, Universidad de Oriente, Santiago de Cuba, Cuba. Emails: mijail@uo.edu.cu, mborges@uo.edu.cu. M. Borges-Quintana has been partially supported by a Post-doctorate scholarship at the University of Valladolid (09-2014 to 02-2015) by Erasmus Mundus Program, Mundus Lindo Project.} \and~Edgar Mart\'inez-Moro
\thanks{E.~Mart\'inez-Moro is with the Institute of Mathematics IMUVa, University of Valladolid. Valladolid, Castilla, Spain. Email: edgar@maf.uva.es. Partially supported by Partially
supported
by
the
Spanish
MINECO
under
grants
MTM2015-65764-C3-1-P
and
MTM2015-69138-REDT.}}
\date{}

\begin{document}



\maketitle






\begin{abstract}
{
In this work we study a weak order ideal associated with the coset leaders of a non-binary linear code. This set allows the incrementally computation of the coset leaders and the definitions of the set of  leader codewords. This set of codewords has some nice properties related to the  monotonicity  of the weight compatible order on the generalized support of a vector in $\mathbb F_q^n$ {
which} allows 
to describe a test set,  a trial set and the set of zero neighbours of a linear code in terms of the leader codewords.}
\end{abstract}

{\bf keywords} Linear codes, {
order ideals}, test set, trial set, zero neighbours, correctable errors


\section{Introduction}

As it is pointed in \cite{klove} it is common folklore in the theory  
of binary linear codes that there is an ordering on the coset leaders chosen as  
the lexicographically smallest minimum weight vectors that provides a  
monotone structure. This is expressed as follows: if  
$\mathbf x$ is a coset leader  and $\mathbf y\subseteq \mathbf x$ (i.e.  
$y_i\leq x_i$ for all $i$)  then $\mathbf y$ is also a coset leader.  
This nice property has been proved of great value, see for example  
\cite{Zemor}, and it has been used for analyzing the error-correction  
capability of binary linear codes \cite{klove}. In this last paper  
the authors introduce  the concept of a \textit{trial set} of codewords  
and they provide a gradient-like decoding algorithm based on this set.

Finding the {
weight} distribution 
of cosets leaders for a code $\mathcal C$ is a classic problem in {
coding theory}. 
This problem is still unsolved for many family of linear codes even for first-order Reed-Muller codes (see \cite{kurshan:1972}).
{
The set of coset leaders is also related with the minimum distance decoding and bounded distance decoding problems as well with   the set of minimal support codewords.  \cite{barg:1998, marquez:2011,massey1993}}

Despite {
their} interest no generalization  of these ideas is known  by the authors of this communication to 
{
non-binary case}. In this paper we provide {
such a} (non straightforward) generalization. 

{
The outline of the paper is as follows Section~\ref{sec:Prel} introduces the idea of a generalized support of a vector 
In Section~\ref{sec:LCodW}  is is defined  the weak order ideal associated with the coset leaders ${\mathcal O}({\mathcal C})$ and it is shown that can be computed incrementally. Theorem~\ref{Theorem4} establishes that all the coset leaders of the code belong to ${\mathcal O}({\mathcal C})$. 
Subsection~\ref{s:LCwords} is devoted to the study of the set of leader codewords of a code as a zero neighbour set and their properties.  
Finally in Section~\ref{Tset-LCW} we analyze the correctable and uncorrectable errors defining a trial set for a linear code from the set of leader codewords.}

The limitations from a practical point view of the results and properties studied in this paper are clear because {
of the size  and the complexity of computing  the set of coset leaders}. Anyway our {
main} interest is   the study and characterizations of some objects related to the codes like zero neighbours, trial set, and the set of correctable and uncorrectable errors. 

\section{Preliminaries}\label{sec:Prel}

From now on we  shall denote by  $\mathbb F_q$ the finite field with 
{
with $q=p^m$ elements, $p$ a prime}. A \emph{linear code} $\mathcal C$ over $\mathbb F_q$ of length $n$ and dimension $k$ 
is a $k$-dimensional subspace of $\mathbb F_q^n$. We will call the vectors $\mathbf v$ in $\mathbb F_q^n$ words and 
 {
 those} $\mathbf v\in \mathcal C$, codewords. For every word $\mathbf v\in \mathbb F_q^n$ its \emph{support} is defined as 
 $\mathrm{supp}(\mathbf v) = \left\{ i \mid v_i \neq 0\right\}$ and its \emph{Hamming weight}, denoted by $\mathrm{w_H}(\mathbf v)$ as the cardinality of $\mathrm{supp}(\mathbf v)$ and the \emph{Hamming distance} $\mathrm{d_H}(\mathbf x, \mathbf y)$ between two words $\mathbf x, ~\mathbf y \in \mathbb F_q^n$ is  
$\mathrm{d_H}(\mathbf x, \mathbf y) = \mathrm{w_H}(\mathbf x - \mathbf y)$. The \emph{minimum distance} $\mathrm{d}(\mathcal C)$ of a linear code $\mathcal C$ is defined as the minimum weight among all nonzero codewords.   
  
The words of minimal Hamming weight in the cosets of $\mathbb F_q^n / \mathcal C$ is the \textit{set of coset leaders} of the code $\mathcal C$ in $\mathbb F_q^n$ and we will denote it by $\mathrm{CL}(\mathcal C)$.  $\mathrm{CL}(\mathbf y)$ will denote the subset of coset leaders corresponding to the coset $\mathbf y+\mathcal C$.  
 Given a coset $\mathbf y+\mathcal C$ we define the \emph{weight of the coset} $\mathrm{w_H}(\mathbf y+\mathcal C)$ as the smallest Hamming weight among all vectors in the coset, or equivalently the weight of one of its leaders. It is well known that {
 given} $t= \lfloor \frac{d(\mathcal C)-1}{2}\rfloor$ 
 where $\lfloor \cdot \rfloor$ denotes the greatest integer function then every coset of weight at most $t$ has a unique coset leader.
  
{
Let $f(X)$ be} an irreducible polynomial over $\mathbb F_p$ of degree $m$ 
{
and} $\beta$ be a root of $f(X)$, then 
any element $a \in \mathbb F_q$ can be represented as 
$a_1 + a_2 \beta + \ldots + a_{m}\beta^{m-1}$ with $
a_i \in \mathbb F_p$ for $i \in \{1, \ldots, m\}$.
For a word $\mathbf v=(v_1, \ldots, v_n) \in \mathbb F_q^n$, such that the $i$-th component of $\mathbf v$ is
$v_i = v_{i,1} + v_{i,2} \beta + \ldots + v_{i,m}\beta^{m-1}$
we define the \emph{generalized support} of a vector $\mathbf v$  as the support of the $nm$-tuple given by the concatenations of the $p$-adic expansion of each component $\mathbf v_i$ of $\mathbf v$, i.e.{
\[\mathrm{supp}_{\mathrm{gen}}(\mathbf v) = (\mathrm{supp}((v_{i,1}, \ldots, v_{i,m})):\, i=1\ldots n),\]
and $\mathrm{supp}_{\mathrm{gen}}(\mathbf v)[i]=\mathrm{supp}((v_{i,1}, \ldots, v_{i,m}))$.}
We will say that $(i,j)\in \mathrm{supp}_{\mathrm{gen}}(\mathbf v)$ if the corresponding $v_{i,j}$ is not zero. From now on the set $$\left\{ \mathbf e_{ij}=\beta^{j-1}\mathbf e_i:\; i = 1, \ldots, n;\, j= 1,\ldots m \right\}$$ will be denoted as $\mathrm{Can}(\mathbb F_q,f)$ and it represents the canonical basis of {$(\mathbb F_q^n,+)$, the additive monoid $\mathbb F_q^n$ with respect to the \lq\lq+\rq\rq$\,$ operation, where $f$ is the irreducible polynomial used to define $\mathbb F_q$}.

We state the following connection between $\mathbb F_q^n$ and $\mathbb N^{nm}$:
\begin{equation*}\begin{split}
\Delta: \mathbb{F}_q^n \to &\, \mathbb N^{nm}\\
{\mathbf v} \mapsto &\, (\psi(v_{i,j}):\,i=1,\ldots,n,\,j=1,\ldots,m),\mbox{ where}\\
\psi: \mathbb F_p\to &\, \mathbb N\\
k \cdot 1_{\mathbb F_p}\mapsto &\, k \,\mbox{mod}\,p.
\end{split}\end{equation*}

On the other hand,
\begin{equation*}\begin{split}
\nabla: \mathbb N^{nm} \to &\, \mathbb{F}_q^n\\
{\mathbf a} \mapsto &\, ((a_{m(i-1)+1} + a_{m(i-1)+2} \beta + \ldots + a_{m(i-1)+m}\beta^{m-1}):\,i=1,\ldots,n).\\
\end{split}\end{equation*}

Given $\mathbf x,\mathbf y\in (\mathbb F_q^n,+)$, $\mathbf x=\sum_{i,j}x_{ij}\mathbf e_{ij}$, $\mathbf y=\sum_{i,j}y_{ij}\mathbf e_{ij}$, we say $\mathbf x\subset \mathbf y$ if $\psi(x_{ij})\leq \psi(y_{ij})$ for all $i\in [1,n]$, and $j\in[1,m]$.

By using $\Delta$ it is possible to relate orders on $\mathbb F_q^n$ with orders on $\mathbb N^{nm}$, and vice versa.  An \textit{admissible order} on $(\mathbb N^{nm},+)$ is a total order $<$ on $\mathbb N^{nm}$ satisfying the following two conditions
\begin{enumerate}
\item $\mathbf 0 <\mathbf x$, for all $\mathbf x\in \mathbb N^{nm},\,\mathbf x\neq \mathbf 0$.
\item If $\mathbf x<\mathbf y$, then $\mathbf x + \mathbf z<\mathbf y + \mathbf z$, for all $\mathbf z\in \mathbb N^{nm}$.
\end{enumerate}
{
In particular, Any admissible order on $({\mathbb N}^{nm},+)$, like the Lexicographical, Degree Lexicographical, Degree Reverse Lexicographical orders, induces an order on $(\mathbb F_q^n,+)$.}

We will say that a representation of a word $\mathbf v$ as an $nm$-tuple over $\mathbb N$ is in \textit{standard form} if $\Delta(\nabla(\mathbf v))=\mathbf v$. 
We will denote the standard form of $\mathbf v$ as $\mathrm{SF}(\mathbf v,f)$ (note that $\nabla(\mathbf v) = \nabla(\mathrm{SF}(\mathbf v,f))$). Therefore, $\mathbf v$ is in standard form if $\mathbf v = \mathrm{SF}(\mathbf v,f)$ (we will also say $\mathbf v \in \mathrm{SF}(\mathbb F_q^n,f)$).

{\begin{remark} From now on we will use  $\mathrm{Can}(\mathbb F_q)$ and $\mathrm{SF}(\mathbb F_q^n)$ instead of $\mathrm{Can}(\mathbb F_q,f)$ and $\mathrm{SF}(\mathbb F_q^n,f)$ respectively, since it is clear that different elections of $f$ or $\beta$ provide equivalent generalized supports.
\end{remark}}


{
\begin{definition}
\label{d:oideal}
A subset ${\mathcal O}$ of $\mathbb N^{k}$ is an order ideal if for all $\mathbf w \in {\mathcal O}$ and $v\in \mathbb N^{k}$ s.t. $\mathbf v_i\leq \mathbf w_i$, $i=1,\ldots, k$, we have $v\in {\mathcal O}$.
\end{definition}
In the same fashion we say that a subset $\mathcal S$ of $\mathbb F_q^n$ is an order ideal if $\Delta({\mathcal S})$ is an order ideal in $\mathbb N^{nm}$.  It is easy to check that an equivalent definition for the order ideal would be that for all $\mathbf w \in \mathcal S$, and for all $(i,j)\in \mathrm{supp}_{\mathrm{gen}}(\mathbf w)$, and $\mathbf v\in \mathbb F_q^n$ s.t. $\mathbf w = \mathbf v + \mathbf e_{ij}$ we have $\mathbf v \in \mathcal S$. If instead of for all $(i,j)\in \mathrm{supp}_{\mathrm{gen}}(\mathbf w)$ the condition is satisfied at least for one $(i,j)\in \mathrm{supp}_{\mathrm{gen}}(\mathbf w)$ we say that ${\cal S}$ is a {\em weak order ideal}.

\begin{definition}
\label{d:woideal}
A subset ${\mathcal S}$ of $\mathbb F_q^{n}$ is a weak order ideal if for all $\mathbf w \in {\mathcal S}\setminus {\mathbf 0}$ there exists $(i,j)\in \mathrm{supp}_{\mathrm{gen}}(\mathbf w)$ s.t. for $\mathbf v\in \mathbb F_q^n$ s.t. $\mathbf w = \mathbf v + \mathbf e_{ij}$ we have $\mathbf v \in \mathcal S$.
\end{definition}
}

\begin{definition}
\label{Voronoi:Definition}
The  \emph{Voronoi region} of a codeword $\mathbf c \in \mathcal C$ 
is the set
\[\mathrm D(\mathbf c) = \left\{ \mathbf y \in \mathbb F_q^n \mid
\mathrm{d_H}(\mathbf y, \mathbf c)\leq \mathrm{d_H}(\mathbf y, \mathbf c'),\,\forall \mathbf c'\in \mathcal C\setminus \{ \mathbf c \} \right\}.\]
\end{definition}
The set of all the Voronoi regions for a given linear code $\mathcal C$ covers the space $\mathbb F_q^n$ and 
$\mathrm D(\mathbf 0) = \mathrm{CL}(\mathcal C)$. However, some words in $\mathbb F_q^n$ may be contained in several regions.
For any subset $A \subset \mathbb F_q^n$ we define $$\mathcal X (A) = \left\{ \mathbf y \in \mathbb F_q^n \mid \min \left\{ \mathrm{d_H}(\mathbf y, \mathbf a):\mathbf a \in A\right\} = 1\right\}$$ as the set of words at Hamming distance $1$ from $A$, i.e.
The \emph{boundary} of $A$ is defined as $\delta (A) = \mathcal X (A) \cup \mathcal X (\mathbb F_q^n \setminus A)$.

\begin{definition}
A nonzero codeword $\mathbf z \in \mathcal C$ is called a \emph{zero neighbour} if its Voronoi region shares a common boundary with the set of coset leaders, i.e. 
$\delta (\mathrm D(\mathbf z)) \cap \delta (\mathrm D(\mathbf 0)) \neq \emptyset.$
The set of all zero neighbours of $\mathcal C$ is denoted by $\mathcal Z(\mathcal C)= \left\{ \mathbf z \in \mathcal C \setminus \{ \mathbf 0\}
~:~ \delta (\mathrm D(\mathbf z)) \cap \delta (\mathrm D(\mathbf 0)) \neq \emptyset
\right\}.$
\end{definition}

\begin{definition}
A \emph{test-set} $\mathcal T$ for a given linear code $\mathcal C$ is a set of codewords such that every word $\mathbf y$ \begin{enumerate}
  \item either $\mathbf y$ lies in $\mathrm D(\mathbf 0)$ 
  \item or there exists $\mathbf v\in \mathcal T$ such that $\mathrm{w_H}(\mathbf y- \mathbf v)< \mathrm{w_H}(\mathbf y)$.
\end{enumerate}
\end{definition}

The set of zero neighbours is a test set, also from the set of zero neighbours can be obtained any minimal test set according to the cardinality of the set  \cite{barg:1998}.

\section{{
The weak order ideal of the coset leaders}}
\label{sec:LCodW}

The first idea that allows us to compute incrementally the set of all coset leaders for a linear code was introduced in \cite{borges:2007a}. In that paper we used the additive structure of $\mathbb F_q^n$ with the set of canonical generators $\mathrm{Can}(\mathbb F_q)$. Unfortunately in \cite{borges:2007a}  most of the chosen coset representatives may not be coset leaders if the weight of the coset is greater than $t$.

\begin{theorem}[\cite{huffman:2003}{, Theorem 1.12.6.v}]
\label{TheoremH}
Assume that $\mathbf x$ is a coset leader of $\mathcal C$. If $\mathbf x^\prime \in \mathbb F_q^n$ and $\mathbf x_i^\prime=\mathbf x_i$ for all $i \in \mathrm{supp}(\mathbf x^\prime)$, then $\mathbf x^\prime$ is also a coset leader of $\mathcal C$.
\end{theorem}

In order to incrementally generate all coset leaders  starting  from $\mathbf 0$ 
adding elements in $\mathrm{Can}(\mathbb F_q)$, we must consider words with weight one more than the previous chosen coset leader. Next result is a byproduct of Theorem~\ref{TheoremH},  we may characterize which vectors we need to generate with weight one more than its coset leader in order to ensure  all coset leaders are generated.
\begin{theorem}
\label{Theorem1} Let $\mathbf x\in \mathrm{SF}(\mathbb F_q^n)$ be an element in  $\mathrm{CL}(\mathcal C)$, let $i \in \mathrm{supp}(\mathbf x)$. If $\mathbf x^\prime \in \mathbb F_q^n$ and $\mathbf x_j^\prime=\mathbf x_j$ for all $j \in \mathrm{supp}(\mathbf x^\prime)\setminus \{i\}$, then $\mathrm{w_H}(\mathbf x^\prime)\leq \mathrm{w_H}(\mathbf x^\prime+{\mathcal C})+1$.
\end{theorem}

\begin{proof}
By Theorem~\ref{TheoremH}, $\mathbf x-\mathbf x_i\in \mathrm{CL}(\mathcal C)$. The proof of the Theorem is analogous to the proof of Theorem 1.12.6.v in \cite{huffman:2003}. Note that if we suppose that $\mathrm{w_H}(\mathbf x^\prime) \geq \mathrm{w_H}(\mathbf x^\prime+{\mathcal C})+2$ it would imply that $\mathbf x$ is not a coset leader, which is a contradiction.
\end{proof}

Let $\mathbf w\in \mathrm{SF}(\mathbb F_q^n)$ be an element in  $\mathrm{CL}(\mathcal C)$, and $(i,j)\in \mathrm{supp_{gen}}(\mathbf w)$. Let $\mathbf y\in \mathrm{SF}(\mathbb F_q^n)$ s.t. $\mathbf w = \mathbf y + \mathbf e_{ij}$ then, as a consequence of the previous theorem we have that 
\begin{equation}
{\mathrm{w_H}(\mathbf y) \leq \mathrm{w_H}(\mathbf y+\mathcal C)+1}.
\end{equation}

In the situation above we will say that the coset leader $\mathbf w$ is an \textit{ancestor} of the word $\mathbf y$, and that $\mathbf y$ is a \textit{descendant} of $\mathbf w$. In the binary case this definitions behave as the ones in \cite[\S11.7]{huffman:2003} but in the case $q\neq 2$ there is a subtle difference, a coset leader could be an ancestor of another coset leader or an ancestor of a word at Hamming distance $1$ to a coset leader (this last case is not possible in the binary case).

\subsection{The set ${\mathcal O}({\mathcal C})$}

{Given $\prec_1$ an admissible order on $(\mathbb N^{nm},+)$ we define the \textit{ weight compatible order} $\prec$ on  $(\mathbb F_q^n,+)$} associated to  $\prec_1$  as the ordering given by
\begin{enumerate}
  \item  $\mathbf x \prec \mathbf y$ if $\mathrm{w_H}(\mathbf x) < \mathrm{w_H}( \mathbf y)$ or \item if $\mathrm{w_H}( \mathbf x)= \mathrm{w_H}( \mathbf y)$ then $\Delta(\mathbf x) \prec_1 \Delta(\mathbf y)$.
\end{enumerate}
I.e. the words are ordered according their weights and the order  $\prec_1$ break ties. These class of orders is a subset of the class of monotone $\alpha$-orderings in \cite{klove}. In fact we will need a little more than monotonicity, for the purpose of this work we will also need that for every pair {$\mathbf v, \mathbf w\in \mathrm{SF}(\mathbb F_q^n)$
\begin{equation}
\label{eq:subpal}
\mbox{ if } \mathbf v\subset \mathbf w, \mbox{ then } \mathbf v \prec \mathbf w.
\end{equation}
}

Note that (\ref{eq:subpal}) is satisfied for a weight compatible order. In addition, for any weight compatible order $\prec$  every strictly decreasing sequence  terminates (due to the finiteness of the set $\mathbb F_q^n$).
{
In the binary case the behavior of the coset leaders 
can be translated to the fact that the set of coset leader is an order ideal of $\mathbb F_2^n$; whereas, for non binary linear codes this is no longer true even if we try to use the characterization of order ideals given in \cite{Braun}, where order ideals do not need to be associated with admissible orders.
}
\begin{definition}
\label{List-Def-LC}
We define  {\em the weak order ideal of the coset leaders} of ${\mathcal C}$ as the set ${\mathcal O}({\mathcal C})$ of elements in $\mathbb F_q^n$ verifying one of the following items:
\begin{enumerate}
\item $\mathbf 0\in {\mathcal O}({\mathcal C})$.
\item Criterion 1: If $\mathbf v \in {\mathcal O}({\mathcal C})$ and $\mathrm{w}_H(\mathbf v) = \mathrm{w}_H\left(\mathbf v+ \mathcal C\right)$ then

\centerline{$\left\{ \mathbf v + \mathbf e_{ij} \mid \Delta(\mathbf v) + \Delta(\mathbf e_{ij})\in \mathrm{SF}(\mathbb F_q^n) \right\} \subset {\mathcal O}({\mathcal C})$.} 
\item Criterion 2: If $\mathbf v \in {\mathcal O}({\mathcal C})$ and $\mathrm{w}_H(\mathbf v) = \mathrm{w}_H\left(\mathbf v+ \mathcal C\right)+1$ then 

\centerline{$\left\{\mathbf v + \mathbf e_{ij} \mid i \in \mathrm{supp}(\mathbf v),\, \Delta(\mathbf v) + \Delta(\mathbf e_{ij})\in \mathrm{SF}(\mathbb F_q^n) \right.,\left. \mathbf v- \mathbf v_i\in \mathrm{CL}(\mathcal C)\right\} \subset {\mathcal O}({\mathcal C})$.} 
\end{enumerate}
\end{definition}

\begin{remark}
\label{r:woidel}
{
It is clear by the {
Criteria 1 and 2} of the definition above that ${\mathcal O}({\mathcal C})$ is a weak order ideal. }
\end{remark}

\begin{theorem}
\label{Theorem3}
Let $\mathbf w\in \mathbb F_q^n$. If there exists $i\in 1,\ldots,n$ s.t. $\mathbf w-\mathbf w_i \in \mathrm{CL}(\mathcal C)$ then $\mathbf w \in {\mathcal O}({\mathcal C})$.
\end{theorem}
\begin{proof}
We will proceed by induction on $\mathbb F_q^n$ with respect to the order $\prec$.
The statement is true for $\mathbf 0 \in \mathbb F_q^n$. Now for the inductive step, 
we assume that the desired property is true for any word $\mathbf u \in \mathbb F_q^n$ such that  there exists $i\in 1,\ldots,n$ s.t. $\mathbf u-\mathbf u_i \in \mathrm{CL}(\mathcal C)$ and also $\mathbf u$ is smaller than an arbitrary but fixed $\mathbf w\neq \mathbf 0$ with respect to $\prec$ and $\mathbf w-\mathbf w_j \in \mathrm{CL}(\mathcal C)$, for some $j\in 1,\ldots,n$, i.e.
$\hbox{if }\mathbf u-\mathbf u_i \in \mathrm{CL}(\mathcal C), \hbox{ for some } i \in 1,\ldots,n, \hbox{ and } \mathbf u \prec \mathbf w \hbox{ then }\mathbf u \in {\mathcal O}({\mathcal C}).$
We will show that the previous conditions imply that $\mathbf w$ is also in ${\mathcal O}({\mathcal C})$.

Let $\mathbf w=\mathbf v + \mathbf e_{ij}$, with $(i,j)\in \mathrm{supp_{gen}}(\mathbf w)$ then  $\mathbf v\prec \mathbf w$ by (\ref{eq:subpal}).  As $\mathbf w-\mathbf w_i \in \mathrm{CL}(\mathcal C)$, the same is true for $\mathbf v$, i.e, $\mathbf v-\mathbf v_i \in \mathrm{CL}(\mathcal C)$ then by the induction hypothesis we have that $\mathbf v \in {\mathcal O}({\mathcal C})$. By Theorem~\ref{Theorem1}, $\mathrm{w_H}(\mathbf v)\leq \mathrm{w_H}(\mathbf v+{\mathcal C})+1$; therefore, by Criteria 1 or 2 in Definition~\ref{List-Def-LC} it is guaranteed that $\mathbf w\in {\mathcal O}({\mathcal C})$.
\end{proof}

\begin{theorem}
\label{Theorem4}
Let $\mathbf w\in \mathbb F_q^n$ and $\mathbf w\in \mathrm{CL}(\mathcal C)$ then $\mathbf w \in {\mathcal O}({\mathcal C})$.
\end{theorem}
\begin{proof}
Let $(i,j)\in \mathrm{supp_{gen}}(\mathbf w)$, since $\mathbf w\in \mathrm{CL}(\mathcal C)$, by Theorem~\ref{Theorem1}, $\mathbf w-\mathbf w_i \in \mathrm{CL}(\mathcal C)$; then, by Theorem~\ref{Theorem3}, $\mathbf w \in {\mathcal O}({\mathcal C})$.
\end{proof}
{
The previous theorem has been shown that ${\mathcal O}({\mathcal C})$ contains the set of coset leaders of the linear code ${\mathcal C}$.
}

\subsection{{
{Zero neighbours and leader codewords}}}\label{s:LCwords}
\begin{definition}
\label{LCwords}
The set of \emph{leader codewords} of a linear code $\mathcal C$  is defined as
\begin{equation}\nonumber
\mathrm L(\mathcal C) =\left\{ 
\begin{array}{c}
\mathbf v_1 + \mathbf e_{ij} - \mathbf v_2  \in \mathcal C\setminus \{\mathbf 0 \} \mid 
\Delta(\mathbf v_1) + \Delta(\mathbf e_{ij}) \in \mathrm{SF}(\mathbb F_q^n),\\ \mathbf v_2 \in \mathrm{CL}(\mathcal C)
\mbox{ and }\mathbf v_1 - {\mathbf v_1}_i \in \mathrm{CL}(\mathcal C)
\end{array}
\right\}.
\end{equation}
\end{definition}

Note that the definition is a bit more complex that the one for binary codes in  \cite{borges:2014} due to the fact that in the general case not all coset leaders need to be ancestors of coset leaders. The name of leader codewords comes from the fact that one could compute all coset leaders of a corresponding word knowing the set $\mathrm L(\mathcal C)$   {
adapting} \cite[Algorithm~3]{borges:2014}.

\begin{remark}
The algorithm for computing $\mathrm L(\mathcal C)$ is based on the construction of ${\mathcal O}({\mathcal C})$. Theorem~\ref{Theorem3} guarantees that $\mathbf w\in {\mathcal O}({\mathcal C})$ provided that $\mathbf w-\mathbf w_i \in \mathrm{CL}(\mathcal C)$ for some $i$, then the associated set of leader codewords may be computed as $\{\mathbf w - \mathbf v:\,\mathbf w \in {\mathcal O}({\mathcal C}),\, \mathbf w-\mathbf w_i \in \mathrm{CL}(\mathcal C),\, \mathbf v \in \mathrm{CL}(\mathbf w)\mbox{ and } \mathbf v\neq \mathbf w\}$. 
 \end{remark}
 
\begin{theorem}[Properties of  $\mathrm L(\mathcal C)$]
\label{t:LCodW}
Let $\mathcal C$ be a linear code then
\begin{enumerate}
\item $\mathrm L(\mathcal C)$ is a test set for $\mathcal C$.
\item Let $\mathbf w$ be an element in  $\mathrm L(\mathcal C)$ then 
$\mathrm{w_H}(\mathbf w) \leq 2 \rho(\mathcal C) +1$ where $\rho(\mathcal C)$ is the covering radius of the code $\mathcal C$.
\item If $\mathbf w \in \mathrm L(\mathcal C)$ then $ \mathcal X (\mathrm D(\mathbf 0)) \cap (\mathrm D(\mathbf w)\cup \mathcal X \mathrm (D(\mathbf w)))\neq \emptyset.$
\item If $\mathcal X (\mathrm D (\mathbf 0)) \cap \mathrm D(\mathbf w) \neq \emptyset$ then $\mathbf w \in \mathrm L(\mathcal C)$.
\end{enumerate}
\end{theorem}

\begin{proof}
{
{ 1)} Let $\mathbf y \notin \mathrm{CL}(\mathcal C)$ and $\mathrm{supp}(\mathbf y) = \left\{ i_k:\;k=1,\ldots,l\right\}$, $l\leq n$. Let $s$ be such that $1\leq s<l$, $\mathbf v_1=\sum_{k=1}^s{\mathbf y}_{i_k}{\mathbf e}_{i_k}\in \mathrm{CL}(\mathcal C)$, $\mathbf v_1 + \mathbf y_{i_{S+1}}\mathbf e_{i_{S+1}}\notin \mathrm{CL}(\mathcal C)$ and $\mathbf v_2=\mathrm{CL}(\mathbf v_1 + \mathbf y_{i_{S+1}}\mathbf e_{i_{S+1}})$. Let $\,\mathbf y_{i_{S+1}}\mathbf e_{i_{S+1}}=\sum_{t=1}^z\mathbf e_{ij_t}$ and ${\mathbf v_1}^\prime=\mathbf v_1+\mathbf y_{i_{S+1}}\mathbf e_{i_{S+1}}-\mathbf e_{ij_Z}$. Then ${\mathbf v_1}^\prime- {\mathbf v_1}^\prime_i=\mathbf v_1\in \mathrm{CL}(\mathcal C)$ implies that $\mathbf w={\mathbf v_1}^\prime + \mathbf e_{ij_Z} - \mathbf v_2\in \mathrm L(\mathcal C)$. In addition,

\noindent $\mathrm{w_H}(\mathbf y - \mathbf w)=\mathrm{w_H}(\mathbf v_2+(\mathbf y - {\mathbf v_1}^\prime - \mathbf e_{ij_Z}))\leq \mathrm{w_H}(\mathbf v_2)+\mathrm{w_H}(\mathbf y - {\mathbf v_1}^\prime - \mathbf e_{ij_Z})$ and 
$\mathrm{w_H}(\mathbf v_2)+\mathrm{w_H}(\mathbf y - {\mathbf v_1}^\prime - \mathbf e_{ij_Z}) < \mathrm{w_H}({\mathbf v_1}^\prime + \mathbf e_{ij_Z}) + \mathrm{w_H}(\mathbf y - {\mathbf v_1}^\prime - \mathbf e_{ij_Z})$.

Note that $\mathrm{supp}({\mathbf v_1}^\prime + \mathbf e_{ij_Z}) \cap \mathrm{supp}(\mathbf y - {\mathbf v_1}^\prime - \mathbf e_{ij_Z})=\emptyset$; consequently, $\mathrm{w_H}(\mathbf y)=\mathrm{w_H}({\mathbf v_1}^\prime + \mathbf e_{ij_Z}) + \mathrm{w_H}(\mathbf y - {\mathbf v_1}^\prime - \mathbf e_{ij_Z})$ and $\mathrm{w_H}(\mathbf y - \mathbf w)<\mathrm{w_H}(\mathbf y)$. Thus, $\mathrm L(\mathcal C)$ is a test set.

{ 2)} Let $\mathbf c \in \mathrm L(\mathcal C)$ then there exists $\mathbf v_2 \in \mathrm{CL}(\mathcal C)$, $\mathbf v_1\in \mathrm{SF}(\mathbb F_q^n)$, $1\leq i\leq n$ and $1\leq j\leq m$ such that $\mathbf v_1 - {\mathbf v_1}_i \in \mathrm{CL}(\mathcal C)$ and $\mathbf c = \mathbf v_1 + \mathbf e_{ij} - \mathbf v_2$. Applying the definition of covering radius we have that $\mathrm{w}_H(\mathbf v_1 - {\mathbf v_1}_i), \mathrm{w}_H(\mathbf v_2) \leq \rho$, thus $\mathrm{w}_H(\mathbf c) \leq 2\rho +1$.

{ 3)} Let $\mathbf w\in \mathrm L(\mathcal C)$, then $\mathbf w=\mathbf v_1 + \mathbf e_{ij} - \mathbf v_2$, where $\mathbf v_1,\, \mathbf v_2$ are elements in $\mathbb F_q^n$ such that $\mathbf v_1 + \mathbf e_{ij} \in \mathrm{SF}(\mathbb F_q^n)$, $\mathbf v_2 \in \mathrm{CL}(\mathcal C)$ 
and $\mathbf v_1 - {\mathbf v_1}_i\in \mathrm{CL}(\mathcal C)$. 

\begin{itemize}
  \item If $\mathbf v_1 + \mathbf e_{ij} \notin \mathrm{CL}(\mathcal C)$, then $\mathbf v_1 + \mathbf e_{ij} \in \mathcal X (\mathrm D(\mathbf 0))$ and $(\mathbf v_1 + \mathbf e_{ij}) -\mathbf w =\mathbf v_2\in \mathrm{CL}(\mathcal C)$ implies that $\mathbf v_1 + \mathbf e_{ij} \in \mathrm D(\mathbf w)$. 

\item If $\mathbf v_1 + \mathbf e_{ij} \in \mathrm{CL}(\mathcal C)$ we define  ${\mathbf v_1^\prime}={\mathbf v_1^\prime}(0)=\mathbf v_1 + \mathbf e_{ij}$. It is clear that  ${\mathbf v_1^\prime},\,\mathbf v_2\in \mathrm{CL}(\mathbf v_2)$. Since $\mathbf w\neq \mathbf 0$ let $l$ be {a number in the set}  $\{1,\ldots,  n\}$ such that
{${\mathbf v_1^\prime}_l - {\mathbf v_2}_l\neq 0$}. Let ${\mathbf v_2}_l=\sum_{j=1}^T\,\mathbf e_{li_j}$, for $1\leq h\leq T$, $\mathbf {\mathbf v_1}^\prime(h)=\mathbf v_1^\prime + \sum_{j=1}^h\,\mathbf e_{li_j}$ and ${\mathbf v_2}(h)=\mathbf v_2 - \sum_{j=1}^h\,\mathbf e_{li_j}$. 

If there exists an $h$ ($1\leq h<T$) such that ${\mathbf v_1^\prime}(h)\notin \mathrm{CL}(\mathcal C)$ and ${\mathbf v_1^\prime}({h-1})\in \mathrm{CL}(\mathcal C)$ then these two conditions imply that 
${\mathbf v_1^\prime}(h)\in \mathcal X (\mathrm D(\mathbf 0))$. On the other hand, ${\mathbf v_2}(h)= {\mathbf v_1^\prime}(h) - \mathbf w$ is either  a coset leader (${\mathbf v_1^\prime}(h)\in \mathrm D(\mathbf w)$) or  $\mathrm{d_H}({\mathbf v_2}(h),{\mathbf v_2})=1$ (${\mathbf v_1^\prime}(h)\in \mathcal X (\mathrm D(\mathbf w))$). 

If there is no such and $h$ ($1\leq h<T$) satisfying the condition then $\mathrm{w_H}({\mathbf v_1^\prime}(T)) = {\mathrm w_H}({\mathbf v_2}(T))+1$, which means that ${\mathbf v_1^\prime}(T)$ is not a coset leader and ${\mathbf v_1^\prime}({T-1})$ is a coset leader. Then using the same idea of the previous paragraph we have that ${\mathbf v_1^\prime}(T)\in \mathcal X (\mathrm D(\mathbf 0))$ and ${\mathbf v_1^\prime}(T)\in \mathrm D(\mathbf w)\cup \mathcal X (\mathrm D(\mathbf w))$. 
\end{itemize}
}

{4) If $\mathcal X (\mathrm D (\mathbf 0)) \cap \mathrm D(\mathbf w) \neq \emptyset$, let $\mathbf u\in \mathcal X (\mathrm D (\mathbf 0)) \cap \mathrm D(\mathbf w)$. The first condition  $\mathbf u\in \mathcal X (\mathrm D (\mathbf 0))$ implies $\mathbf u=\mathbf v_1 +\mathbf e_{ij}$ for some $\mathbf v_1\in \mathrm{SF}(\mathbb F_q^n)$ 
, $(i,j)
\in \mathrm{supp_{gen}}(\mathbf u)$ and $\mathbf v_1 - {\mathbf v_1}_i\in \mathrm{CL}(\mathcal C)$. 
On the other hand, $\mathbf v_1+\mathbf e_{ij} \in \mathrm D(\mathbf w)$ implies that $\mathbf v_2=(\mathbf v_1 + \mathbf e_{ij}) - \mathbf w\in \mathrm{CL}(\mathcal C)$. Therefore, $\mathbf w= \mathbf v_1 + \mathbf e_{ij}- \mathbf v_2\in \mathrm L(\mathcal C)$.}
\end{proof} 

\begin{remark}
Note that item {\it 3} in Theorem~\ref{t:LCodW} implies that any leader codeword is a zero neighbour however, one of the differences with the binary case is that it is not always true that for a leader codeword $\mathbf w$ we have that $\mathcal X (\mathrm D (\mathbf 0)) \cap \mathrm D(\mathbf w) \neq \emptyset$, although by item {\it 4} we have that $\mathbf w$ is a leader codeword provided this condition is satisfied. Furthermore, item {\it 4} guarantees that the set of leader codewords contains all the minimal test set according to its cardinality (see \cite{barg:1998}). As a consequence of all these properties in Theorem~\ref{t:LCodW} we could say that the the set of leader codewords is a ``good enough'' subset  of the set of zero neighbours.
\end{remark}

\section{Correctable and uncorrectable errors}
\label{Tset-LCW}

We define the relation $\subset_1$ in the additive monoid which describe exactly the relation $\subset$ in the vector space $\mathbb F_q^n$. Given $\mathbf x,\mathbf y\in (\mathbb F_q^n,+)$
\begin{equation}
\label{eq:subset1}
 \mathbf x\subset_1 \mathbf y \hbox{ if } \mathbf x\subset \mathbf y \hbox{ and } \mathrm{supp}(\mathbf x)\cap \mathrm{supp}(\mathbf y-\mathbf x)=\emptyset.
\end{equation}

Note that this definition translates to $\mathbb F_q^n$ the binary case situation in \cite{klove}. In this case given a $\mathbf y\in (\mathbb F_q^n,+)$ { there are more words $x\in (\mathbb F_q^n,+)$  such that $\mathbf x\subset \mathbf y$ than if we consider $\mathbf x, \mathbf y$ as elements in the vector space $\mathbb F_q^n$}. Of course, any relation $\mathbf x\subset \mathbf y$ in $\mathbb F_q^n$ as a vector space it is also true in the additive monoid, {but it is not true the other way round}.

The set ${E}^0({\mathcal C})$ of \textit{correctable errors} of a linear code $\mathcal C$ is the set of  the minimal { elements with respect to}  $\prec$ in each coset. The elements of the set ${E}^1({\mathcal C})=\mathbb F_q^n\setminus {E}^0({\mathcal C})$ will be  called \textit{uncorrectable errors}. A {\em trial set} $T\subset {\mathcal C}\setminus \{{\mathbf 0}\}$ of the code ${\mathcal C}$ is a set which has the following property
\[
\label{eq:trialset}
{\mathbf y}\in {E}^0({\mathcal C}) \mbox{ if and only if } {\mathbf y} \preceq {\mathbf y} + {\mathbf c}, \mbox{ for all } {\mathbf c}\in T.\]

Since $\prec$ is a monotone $\alpha$-ordering on $\mathbb F_q^n$, the set of correctable and uncorrectable errors form a monotone structure. Namely, if $\mathbf x\subset_1 \mathbf y$ then $\mathbf x \in {E}^1({\mathcal C})$ implies $\mathbf y \in {E}^1({\mathcal C})$ and $\mathbf y \in {E}^0({\mathcal C})$ implies $\mathbf x \in {E}^0({\mathcal C})$. In the {general case $q\neq 2$ there is} a  difference with respect to the binary case, there may be words $\mathbf y^\prime\in \mathbb F_q^n$ s.t. $\mathrm{supp_{gen}}(\mathbf y^\prime)=\mathrm{supp_{gen}}(\mathbf y)$, {$\mathbf y^\prime\subset \mathbf y$} and $\mathbf y^\prime$ could be either a correctable error or an uncorrectable error, so, the monotone structure it is not sustained by $\subset$ in the additive monoid $(\mathbb F_q^n,+)$. 

Let the set of minimal uncorrectable errors $M^1({\mathcal C})$ be the set of $\mathbf y\in {E}^1({\mathcal C})$ such that, if $\mathbf x \subseteq_1 \mathbf y$ and $\mathbf x\in {E}^1({\mathcal C})$, then $\mathbf x= \mathbf y$. In a similar way, the set of maximal correctable errors is the set $M^0({\mathcal C})$ of elements $\mathbf x\in {E}^0({\mathcal C})$ such that, if $\mathbf x\subseteq_1 \mathbf y$ and $\mathbf y\in {E}^0({\mathcal C})$, then $\mathbf x= \mathbf y$.

For $\mathbf c\in {\mathcal C}\setminus \{\mathbf 0\}$, a {\em larger half} is defined as a minimal word $\mathbf u \subseteq_1 \mathbf c$ in the ordering $\preceq$ such that $\mathbf u - \mathbf c \prec \mathbf u$. The weight of such a word $\mathbf u$ is such that \[\mathrm{w_H}(\mathbf c)\leq 2\mathrm{w_H}(\mathbf u)\leq \mathrm{w}_H(\mathbf c)+ 2,\] see \cite{klove} for more details. The set of larger halves of a codeword $\mathbf c$ is denoted by $\mathrm{L_H}(\mathbf c)$, and for $U\subseteq {\mathcal C}\setminus \{\mathbf 0\}$ the set of larger halves for elements of $U$ is denoted by $\mathrm{L_H}(U)$. Note that $\mathrm{L_H}({\mathcal C})\subseteq {E}^1({\mathcal C})$.

For any ${\mathbf y\in \mathbb F_q^n}$, let $H(\mathbf y)=\{c\in {\mathcal C:\, \mathbf y- \mathbf c \prec \mathbf y}\}$, and we have $\mathbf y\in {E}^0({\mathcal C})$ if and only if $H(\mathbf y)=\emptyset$, and $\mathbf y\in {E}^1({\mathcal C})$ if and only if $H(\mathbf y)\neq \emptyset$. 

In \cite[Theorem~1]{klove}  there is a characterization of the set $M^1({\mathcal C})$ in terms of $H(\cdot)$ and larger halves of the set of minimal codewords $M({\mathcal C})$ for the binary case. It is easy to proof that this Theorem and    \cite[Corollary~3]{klove} are also true for any linear code.

\begin{proposition}[Corollary 3 in \cite{klove} ]
\label{prop:klove}
Let ${\mathcal C}$ be a linear code and $T\subseteq {\mathcal C}\setminus \{\mathbf 0\}$. The following statements are equivalent,:
\begin{enumerate}
\item $T$ is a trial set for ${\mathcal C}$.
\item If $\mathbf y \in M^1({\mathcal C})$, then $T\cap H(\mathbf y)\neq \emptyset$.
\item $M^1({\mathcal C})\subseteq \mathrm{L_H}(T)$.
\end{enumerate}
\end{proposition}

Now we will formulate the result which relates the trial sets {for a given weight compatible order $\prec$} and the set of leader codewords.
\begin{theorem}
\label{teo:ts-LCW}
Let $\mathcal C$ be a linear code and $\mathrm L(\mathcal C)$ the set of leaders codewords for $\mathcal C$, the following statements are satisfied.
\begin{enumerate}
\item $\mathrm L(\mathcal C)$ is a trial set for any given $\prec$.
\item Algorithm~2 in \cite{borges:2014} can be adapted to compute a set of leader codewords which is a trial set $T$ for a given $\prec$ such that satisfies the following property
\begin{description}
\item[] For any $\mathbf c\in T$, there exists $\mathbf y\in M^1({\mathcal C})\cap \mathrm{L_H}(\mathbf c)$ s.t. $\mathbf y- \mathbf c\in {E}^0({\mathcal C})$.
\end{description}
\end{enumerate}

\end{theorem}
\begin{proof}$\,$\\
{Proof of 1) We will prove statement 2 of Proposition~\ref{prop:klove}. Let $\mathbf y \in M^1({\mathcal C})$, let $i$ such that $\mathrm{supp_{gen}}(\mathbf y)[i]\neq \emptyset$ and $\mathbf v_1=\mathbf y -\mathbf y_i$. Since $\mathbf y \in M^1({\mathcal C})$ we have that $\mathbf v_1\in {E}^0({\mathcal C})$, thus it is a coset leader. On the other hand, let $\mathbf v_2\in {E}^0({\mathcal C})$ such that $\mathbf v_2 \in \mathrm{CL}(\mathbf y)$ and $\mathbf c=\mathbf y - \mathbf v_2$. It is clear that $\mathbf c$ is a leader codeword and $\mathbf y - \mathbf c =\mathbf v_2\prec \mathbf y$. Therefore  $\mathbf c\in H(\mathbf y)$.
}

\noindent {Proof of 2) }In the Algorithm~2 in \cite{borges:2014}, as the first step, it is necessary to add to the function ${\tt InsertNext}$ the Criteria 2 of the construction of ${\mathcal O}({\mathcal C})$, whose elements are stored in ${\tt Listing}$. On the other hand, in the steps of the construction of the leader codewords (Steps 11 - 13) it is enough to state the condition $\mathbf t\in M^1({\mathcal C})$ for $\mathbf t$ and taking $\mathbf t_k$ only equal to the coset leader of $\mathrm{CL}(\mathbf t)$, that is the corresponding correctable error and add the codeword $\mathbf t-\mathbf t_k$ to the set $\mathrm L(\mathcal C)$.

\end{proof}

\end{document}